\documentclass[12pt]{article}      
\usepackage{amsmath,amssymb,amsthm, amscd, graphicx, verbatim, setspace} 
\usepackage{setspace}

\theoremstyle{plain}
\newtheorem{thm}{Theorem}
\newtheorem{lem}[thm]{Lemma}
\newtheorem{cor}[thm]{Corollary}

\DeclareMathOperator{\ex}{ex}
\DeclareMathOperator{\Ex}{Ex}

\title{Bounds on parameters of minimally non-linear patterns}
\date{}
\author{P.A. CrowdMath\\
\small\tt crowdmath@artofproblemsolving.com}
\begin{document}
\maketitle

\begin{abstract}
Let $\ex(n, P)$ be the maximum possible number of ones in any 0-1 matrix of dimensions $n \times n$ that avoids $P$. Matrix $P$ is called minimally non-linear if $\ex(n, P) = \omega(n)$ but $\ex(n, P') = O(n)$ for every strict subpattern $P'$ of $P$. We prove that the ratio between the length and width of any minimally non-linear 0-1 matrix is at most $4$, and that a minimally non-linear 0-1 matrix with $k$ rows has at most $5k-3$ ones. We also obtain an upper bound on the number of minimally non-linear 0-1 matrices with $k$ rows. 

In addition, we prove corresponding bounds for minimally non-linear ordered graphs. The minimal non-linearity that we investigate for ordered graphs is for the extremal function $\ex_{<}(n, G)$, which is the maximum possible number of edges in any ordered graph on $n$ vertices with no ordered subgraph isomorphic to $G$.
\end{abstract}

\section{Introduction}
A 0-1 matrix $M$ \emph{contains} a 0-1 matrix $P$ if some submatrix of $M$ either equals $P$ or can be turned into $P$ by changing some ones to zeroes. Otherwise $M$ \emph{avoids} $P$. The function $\ex(n, P)$ is the maximum number of ones in any 0-1 matrix of dimensions $n \times n$ that avoids $P$.

The function $\ex(n, P)$ has been used for many applications, including resolving the Stanley-Wilf conjecture \cite{MT} and bounding the maximum number of unit distances in a convex $n$-gon \cite{Fure}, the complexity of algorithms for minimizing rectilinear path distance while avoiding obstacles \cite{Mit}, the maximum number of edges in ordered graphs on $n$ vertices avoiding fixed ordered subgraphs \cite{KM, PT, weid}, and the maximum lengths of sequences that avoid certain subsequences \cite{Pet}.

It is easy to see that $\ex(n, P) = \ex(n, P')$ if $P'$ is obtained from $P$ by reflections over horizontal or vertical lines or ninety-degree rotations. It is also obvious that if $P'$ contains $P$, then $\ex(n, P) \leq \ex(n, P')$.

If $P$ has at least two entries and at least one $1$-entry, then $\ex(n, P) \geq n$ since there exists a matrix with ones only in a single column or a single row that avoids $P$. For example, $\ex(n, \begin{bmatrix}
1 & 1
\end{bmatrix}) = n$ since the $n \times n$ matrix with ones only in the first column and zeroes elsewhere avoids $\begin{bmatrix}
1 & 1
\end{bmatrix}$ and every 0-1 matrix of dimensions $n \times n$ with $n+1$ ones has a row with at least two ones. It is also easy to see that $\ex(n, P) = (k-1)n$ when $P$ is a $1 \times k$ matrix with all ones: the $n \times n$ matrix with ones only in the first $k-1$ columns and zeroes elsewhere avoids $P$, while every 0-1 matrix with dimensions $n \times n$ and $(k-1)n+1$ ones has a row with at least $k$ ones.

Since the 0-1 matrix extremal function has a linear lower bound for all 0-1 matrices except those with all zeroes or just one entry, it is natural to ask which 0-1 matrices have linear upper bounds on their extremal functions. F\"{u}redi and Hajnal posed the problem of finding all 0-1 matrices $P$ such that $\ex(n, P) = O(n)$ \cite{FH}. Their problem has only been partially answered. 

Marcus and Tardos proved that $\ex(n, P) = O(n)$ for every permutation matrix $P$ \cite{MT}. This linear bound was extended in \cite{G} to tuple permutation matrices, which are obtained by replacing every column of a permutation matrix with multiple adjacent copies of itself. 

Keszegh \cite{Ke}, Tardos \cite{gtar} and F\"{u}redi and Hajnal \cite{FH} found multiple operations that can be used to construct new linear 0-1 matrices (matrices $P$ for which $\ex(n, P) = O(n)$) from known linear 0-1 matrices. No one has found a way to determine whether an arbitrary  0-1 matrix is linear just by looking at it. However, one approach that might eventually resolve the F\"{u}redi-Hajnal problem is to identify all minimally non-linear 0-1 matrices.

A 0-1 matrix $P$ is called \emph{minimally non-linear} if $\ex(n, P) = \omega(n)$ but $\ex(n, P') = O(n)$ for every $P'$ that is strictly contained in $P$. If $M$ contains a minimally non-linear 0-1 matrix, then $\ex(n, M)$ is non-linear. If $M$ avoids all minimally non-linear 0-1 matrices, then $\ex(n, M)$ is linear. Thus identifying all minimally non-linear 0-1 matrices is equivalent to solving F\"{u}redi and Hajnal's problem.

Keszegh \cite{Ke} constructed a class $H_{k}$ of 0-1 matrices for which $\ex(n, H_{k}) = \Theta(n \log{n})$ and conjectured the existence of infinitely many minimally non-linear 0-1 matrices contained in the class. This conjecture was confirmed in \cite{G}, without actually constructing an infinite family of minimally non-linear 0-1 matrices.

There are only seven minimally non-linear 0-1 matrices with $2$ rows. These matrices include $\begin{bmatrix}
1 & 1\\
1 & 1
\end{bmatrix}, \begin{bmatrix}
1 & 0 & 1 \\
0 & 1 & 1
\end{bmatrix}, \begin{bmatrix}
0 & 1 & 1 \\
1 & 0 & 1
\end{bmatrix}, \begin{bmatrix}
1 & 0 & 1 & 0 \\
0 & 1 & 0 & 1
\end{bmatrix}$, and reflections of the last three over a vertical line.

In this paper, we bound the number of minimally non-linear 0-1 matrices with $k$ rows for $k > 2$. In order to obtain upper bounds for this number, we bound the ratio between the length and width of a minimally non-linear 0-1 matrix. We also investigate similar problems for sequences and ordered graphs. 

In Section \ref{seq}, we bound the lengths as well as the number of minimally non-linear sequences with $k$ distinct letters. These bounds are easier to obtain than the bounds on minimally non-linear 0-1 matrices and ordered graphs, since they rely mainly on the fact that every minimally non-linear sequence not isomorphic to $ababa$ must avoid $ababa$.

In Section \ref{01mat}, we bound the number of minimally non-linear 0-1 matrices with $k$ rows. We also prove that the ratio between the length and width of a minimally non-linear 0-1 matrix is at most $4$ and that a minimally non-linear 0-1 matrix with $k$ rows has at most $5k-3$ ones. In Section \ref{ordg}, we find corresponding bounds for extremal functions of forbidden ordered graphs. 

\section{Minimally non-linear patterns in sequences}\label{seq}

A sequence $u$ contains a sequence $v$ if some subsequence of $u$ is isomorphic to $v$. Otherwise $u$ avoids $v$. If $u$ has $r$ distinct letters, then the function $\Ex(u, n)$ is the maximum possible length of any sequence that avoids $u$ with $n$ distinct letters in which every $r$ consecutive letters are distinct. 

Like the extremal function $\ex(n, P)$ for forbidden 0-1 matrices, $\Ex(u, n)$ has been used for many applications in combinatorics and computational geometry. These applications include upper bounds on the complexity of lower envelopes of sets of polynomials of bounded degree \cite{3}, the complexity of faces in arrangements of arcs with a limited number of crossings \cite{1}, and the maximum possible number of edges in $k$-quasiplanar graphs on $n$ vertices with no pair of edges intersecting in more than $t$ points \cite{5, 10}.

Minimal non-linearity for $\Ex(u, n)$ is defined as for $\ex(n, P)$. Only the sequences equivalent to $ababa$, $abcacbc$, or its reversal are currently known to be minimally non-linear, but a few other minimally non-linear sequences are known to exist \cite{petmnl}.

In order to bound the number of minimally non-linear sequences with $k$ distinct letters, we bound the length of such sequences in terms of the extremal function $\Ex(a b a b a, k)$, which satisfies $\Ex(ababa,k) \sim 2k\alpha(k)$ \cite{7, 8}.

In the next proof, we use a well-known fact about the function $\Ex(u, n)$ \cite{7}: If $u$ is a linear sequence and $u'$ is obtained from $u$ by inserting the letter $a$ between two adjacent occurrences of $a$ in $u$, then $u'$ is linear.

\begin{lem}
The maximum possible length of a minimally non-linear sequence with $k$ distinct letters is at most $2 \Ex(a b a b a, k)$.
\end{lem}

\begin{proof}
First we claim that there is no immediate repetition of letters greater than $2$ in a minimally non-linear sequence. Suppose for contradiction that there is a minimally non-linear sequence $u$ with a repetition of length at least $3$. 

Remove one of the letters in the repetition and get $u'$. By definition $u'$ is linear, but then inserting the letter back still gives a linear sequence by the well-known fact stated before this lemma, a contradiction.

If $u$ is not isomorphic to $ababa$, then the number of segments of repeated letters in $u$ is at most $\Ex(ababa, k)$ because $u$ avoids $ababa$. Thus $u$ has length at most $2 \Ex(ababa, k)$ since each segment has length at most $2$.
\end{proof}

\begin{cor}
The number of minimally non-linear sequences with $k$ distinct letters is at most $2k \sum_{i=1}^{\Ex(ababa,k)}(2k-2)^{i-1}$.
\end{cor}

\begin{proof}
The number of segments of repeated letters is at most $\Ex(ababa, k)$. Each segment can be filled with one of at most $k$ letters, with length $1$ or $2$, with no adjacent segments having the same letters. 

Thus there are at most $2k$ choices for the first segment and at most $2k-2$ choices for the remaining segments. So the number of such sequences is bounded by $2k \sum_{i=1}^{\Ex(ababa,k)}(2k-2)^{i-1}$.
\end{proof}

\section{Minimally non-linear patterns in 0-1 matrices}\label{01mat}

Although the existence of infinitely many minimally non-linear 0-1 matrices was proved in \cite{G}, only finitely many minimally non-linear 0-1 matrices have been identified. It is an open problem to identify an infinite family of minimally non-linear 0-1 matrices.

In this section, we prove an upper bound of $\sum_{i=\lceil (k+2)/4
\rceil}^{4k-2}(i^{k}-(i-1)^{k}) k^{i-1}$ on the number of minimally non-linear 0-1 matrices with $k$ rows. In order to obtain this bound, we first show that any minimally non-linear 0-1 matrix with $k$ rows has at most $4k-2$ columns. Next, we bound the number of minimally non-linear 0-1 matrices with $k$ rows and $c$ columns. We prove this bound by showing that no column of a minimally non-linear 0-1 matrix has multiple ones after leftmost ones are removed from each row, unless the matrix is the $2 \times 2$ matrix of all ones, $\begin{bmatrix}
1 & 0 & 1 \\
0 & 1 & 1
\end{bmatrix}$, or its reflection over a horizontal line.

In order to bound the ratio between the length and width of any minimally non-linear 0-1 matrix, we use a few well-known lemmas about 0-1 matrix extremal functions. These facts are proved in \cite{FH, gtar}.

\begin{lem}
\begin{enumerate}
\item If $P$ has two adjacent ones $x$ and $y$ in the same row in columns $c$ and $d$, and $P'$ is obtained from $P$ by inserting a new column between $c$ and $d$ with a single one between $x$ and $y$ and zeroes elsewhere, then $\ex(n, P) \leq \ex(n, P') \leq 2 \ex(n, P)$.
\item If $P'$ is obtained by inserting columns or rows with all zeroes into $P$, then $\ex(n, P') = O(\ex(n, P)+n)$.
\item If $P = \begin{bmatrix}
1 & 0 & 1 & 0 \\
0 & 1 & 0 & 1
\end{bmatrix}$, then $\ex(n, P) = \Theta(n \alpha(n))$, where $\alpha(n)$ denotes the inverse Ackermann function.
\end{enumerate}
\end{lem}

The next theorem shows that a minimally non-linear 0-1 matrix must not be more than four times longer than it is wide. The greatest known ratio between the length and width of a minimally non-linear 0-1 matrix is $2$ for the matrix $\begin{bmatrix}
1 & 0 & 1 & 0 \\
0 & 1 & 0 & 1
\end{bmatrix}$.

\begin{thm}\label{main01}
The ratio of width over height of any minimally non-linear matrix is between $0.25$ and $4$.
\end{thm}

\begin{proof}
Since the lemma holds for $\begin{bmatrix}
1 & 0 & 1 & 0 \\
0 & 1 & 0 & 1
\end{bmatrix}$ and its reflections, suppose that $P$ is a minimally non-linear 0-1 matrix with $k$ rows that is not equal to $\begin{bmatrix}
1 & 0 & 1 & 0 \\
0 & 1 & 0 & 1
\end{bmatrix}$ or its reflections. 

Let $P'$ be obtained by scanning through the columns of $P$ from left to right. The first column of $P'$ has a one only in the first row where the first column of $P$ has a one. For $i > 1$, the $i^{th}$ column of $P'$ has a one only in the first row where the $i^{th}$ column of $P$ has a one and where the $(i-1)^{st}$ column of $P'$ does not have a one, unless the $i^{th}$ column of $P$ only has a single one. If the $i^{th}$ column of $P$ only has a single one in row $r$, then the $i^{th}$ column of $P'$ has a one only in row $r$.

The reduction produces a 0-1 matrix with a single one in each column. Let each of the rows $1, \ldots, k$ of $P$ and $P'$ correspond to a letter $a_{1}, \ldots, a_{k}$, and construct a sequence $S$ from $P'$ so that the $i^{th}$ letter of $S$ is $a_{j}$ if and only if $P'$ has a one in row $j$ and column $i$.

By definition $|S|$ equals the number of columns of the minimally non-linear pattern $P$. There cannot be $3$ adjacent same letters in $S$, because any $3$ adjacent same letters implies a column in $P$ with a single $1$ and the immediate right and left neighbors of the $1$-entry being $1$ as well, which would imply that $P$ is not minimally non-linear. Also $S$ avoids $abab$ because otherwise $P$ contains
$\begin{bmatrix}
1 & 0 & 1 & 0 \\
0 & 1 & 0 & 1
\end{bmatrix}$ or its reflection, which are non-linear. So $|S|\leq 2\Ex(abab, k)=4k-2$. This shows that the ratio of width over height of a minimally non-linear matrix is between $0.25$ and $4$.
\end{proof}

Using the bound that we obtained on the number of columns in a minimally non-linear 0-1 matrix with $k$ rows, next we prove that the number of ones in a minimally non-linear 0-1 matrix with $k$ rows is at most $5k-3$. Note that any minimally non-linear 0-1 matrix with $k$ rows has at least $k$ ones since it has no rows with all zeroes.

In order to bound the number of ones in a minimally non-linear 0-1 matrix with $k$ rows, we first prove a more general bound on the number of ones in a minimally non-linear 0-1 matrix with $k$ rows and $c$ columns, assuming that it is not the $2 \times 2$ matrix of all ones.

\begin{lem}\label{01edge}
The number of ones in any minimally non-linear 0-1 matrix with $k$ rows and $c$ columns, besides the $2 \times 2$ matrix of all ones, is at most $k+c-1$.
\end{lem}

\begin{proof}
The result is true for $Q = \begin{bmatrix}
1 & 0 & 1 \\
0 & 1 & 1
\end{bmatrix}$, so suppose that $P$ is a minimally non-linear 0-1 matrix with $k$ rows that is not equal to $Q$, its reflection $\bar{Q}$ over a horizontal line, or the $2 \times 2$ matrix $R$ of all ones. Then $P$ must avoid $Q$, $\bar{Q}$, and $R$. 

If $P$ has $k$ rows and $c$ columns, then remove the first one in each row to obtain a new matrix $P'$. Matrix $P'$ cannot have any column with multiple ones, since otherwise $P$ would contain $Q$, $\bar{Q}$, or $R$. Thus $P'$ has at most $c-1$ ones since the first column has no ones, so $P$ has at most $k+c-1$ ones. 
\end{proof}

\begin{cor}
The number of ones in any minimally non-linear 0-1 matrix with $k$ rows is at most $5k-3$.
\end{cor}

\begin{proof}
Suppose that the minimally non-linear 0-1 matrix $P$ has $k$ rows and $c$ columns. Since $c \leq 4k-2$, matrix $P$ has at most $5k-3$ ones.
\end{proof}

Using the bound on the number of columns in a minimally non-linear 0-1 matrix with $k$ rows, combined with the technique that we used to bound the number of ones in a minimally non-linear 0-1 matrix with $k$ rows, we prove an upper bound on the number of minimally non-linear 0-1 matrices with $k$ rows.

\begin{cor}
For $k > 2$, the number of minimally non-linear 0-1 matrices with $k$ rows is at most $\sum_{i=\lceil (k+2)/4 \rceil}^{4k-2}(i^{k}-(i-1)^{k})k^{i-1}$.
\end{cor}

\begin{proof}
In a minimally non-linear 0-1 matrix with $k$ rows and $i$ columns, there are at most
$i^k-(i-1)^{k}$ possible combinations of leftmost ones that can be deleted in each row, because
having all leftmost ones in the rightmost $i-1$ columns implies that the first column is empty, which is
impossible. After leftmost ones are deleted in each row, each column except the first has at most a
single one. If a column has no one removed, then it stays non-empty with $k$ possibilities. If a
column has at least a one removed, say in the second row, then it cannot become a column with a
one in the second row. In either case, every column except for the first has at most $k$
possibilities, leaving at most $k^{i-1}$ possible matrices. Moreover there are between $\lceil (k+2)/4 \rceil$ and $4k-2$ columns in a minimally non-linear 0-1 matrix with $k$ rows.
\end{proof}

\section{Minimally non-linear patterns in ordered graphs}\label{ordg}

In this section, we prove bounds on parameters of minimally non-linear ordered graphs. The definitions of avoidance, extremal functions, and minimal non-linearity for ordered graphs are analogous to the corresponding definitions for 0-1 matrices.

If $H$ and $G$ are any ordered graphs, then $H$ avoids $G$ if no subgraph of $H$ is order isomorphic to $G$. The extremal function $\ex_{<}(n, G)$ is the maximum possible number of edges in any ordered graph with $n$ vertices that avoids $G$. 

Past research on $\ex_{<}$ has identified similarities with the $0-1$ matrix extremal function $\ex$. For example, Klazar and Marcus \cite{KM} proved that $\ex_{<}(n, G) = O(n)$ for every ordered bipartite matching $G$ with interval chromatic number $2$. This is analogous to the result of Marcus and Tardos \cite{MT} that $\ex(n, P) = O(n)$ for every permutation matrix $P$. Weidert also identified several parallels between $\ex_{<}$ and $\ex$ \cite{weid}, including linear bounds on extremal functions of forbidden tuple matchings with interval chromatic number $2$. These bounds were analogous to the linear bounds for tuple permutation matrices that were proved in \cite{G}.

In order to prove results about minimally non-linear ordered graphs, we use two lemmas about $\ex_{<}(n, G)$. The first is from \cite{weid}:

\begin{lem}
\cite{weid} If $G'$ is created from $G$ by inserting a single vertex $v$ of degree one between two consecutive vertices that are both adjacent to $v$'s neighbor, then $\ex_{<}(n, G') \leq 2\ex_{<}(n, G)$.
\end{lem}

The second lemma and its proof is by Gabor Tardos via private communication \cite{gtardos}.

\begin{lem}
\cite{gtardos} If $G'$ is an ordered graph obtained from $G$ by adding an edgeless vertex, then $ex_{<}(n, G') = O(ex_{<}(n, G)+n)$.
\end{lem}

\begin{proof}
For simplicity assume the new isolated vertex in $G'$ is neither first nor last. Let $H'$ be an ordered graph avoiding $G'$. Take uniform random sample $R$ of the vertices of $H'$, then select a subset $S$ of $R$ deterministically by throwing away the second vertex from every pair of consecutive vertices in $V(H')$ if both of them were selected in $R$. Now $S$ is a subset of vertices without a consecutive pair, so $H=H'[S]$ avoids $G$, since you can stick in a vertex between any two wherever you wish. Now every edge of $H'$ has a minimum of $1/16$ chance of being in $H$ except the edges connecting neighboring vertices, which have no chance. Thus $w(H')<16E[w(H)]+n$ and we are done.
\end{proof}

Most of the results that we prove in this section about minimal non-linearity for the extremal function $\ex_{<}$ are analogous to the results that we proved in the last section about minimal non-linearity for the 0-1 matrix extremal function $\ex$. First we prove that the number of edges in any minimally non-linear ordered graph with $k$ vertices is at most $2k-2$. Since there are no singleton vertices in a minimally non-linear ordered graph, there is a lower bound of $k/2$ on the number of edges.

\begin{thm}\label{allord}
Any minimally non-linear ordered graph with $k$ vertices has at most $2k-2$ edges.
\end{thm}

\begin{proof}
For a 0-1 matrix $P$, define $G_o\left(P\right)$ to be the family of all bipartite ordered graphs with a unique decomposition into two independent sets that form a 0-1 matrix equivalent to $P$ when the vertices in each set are arranged in either increasing or decreasing order as columns and rows with edges corresponding to ones. Then every element of $G_o\left(\begin{bmatrix}
1 & 0 & 1\\
0 & 1 & 1
\end{bmatrix}\right) \cup G_o\left(\begin{bmatrix}
1 & 1\\
1 & 1
\end{bmatrix}\right)$ is non-linear for $\ex_<$, since $\begin{bmatrix}
1 & 0 & 1\\
0 & 1 & 1
\end{bmatrix}$ and $\begin{bmatrix}
1 & 1\\
1 & 1
\end{bmatrix}$ are non-linear for $\ex$ and any ordered graph with interval chromatic number more than $2$ is non-linear for $\ex_<$ \cite{weid}. 

The lemma is true for every element of $G_o\left(\begin{bmatrix}
1 & 0 & 1\\
0 & 1 & 1
\end{bmatrix}\right) \cup G_o\left(\begin{bmatrix}
1 & 1\\
1 & 1
\end{bmatrix}\right)$, so let $G$ be a minimally non-linear ordered graph that is not equal to any element of $G_o\left(\begin{bmatrix}
1 & 0 & 1\\
0 & 1 & 1
\end{bmatrix}\right) \cup G_o\left(\begin{bmatrix}
1 & 1\\
1 & 1
\end{bmatrix}\right)$. 

Thus $G$ avoids every element of $G_o\left(\begin{bmatrix}
1 & 0 & 1\\
0 & 1 & 1
\end{bmatrix}\right) \cup G_o\left(\begin{bmatrix}
1 & 1\\
1 & 1
\end{bmatrix}\right)$. Define an edge as $e=(v_i,v_j)$ where $v_i<v_j$. Remove all edges $(v_i,v_j)$ where $v_j$ is the smallest number $t$ such that $(v_i,t)\in E(G)$. There are at most $k-1$ such edges.

The resulting graph $G'$ cannot have both edges $(v_i, v_k)$ and $(v_j, v_k)$ for any node $v_k$. Because if it does, then there are $v_i<v_a<v_k$ and $v_j<v_b<v_k$ such that $(v_i,v_a)$, $(v_i,v_k)$, $(v_j,v_b)$, and $(v_j,v_k)$ are all in $E(G)$, and therefore $G$ must contain some element in $G_o\left(\begin{bmatrix}
1 & 0 & 1\\
0 & 1 & 1
\end{bmatrix}\right) \cup G_o\left(\begin{bmatrix}
1 & 1\\
1 & 1
\end{bmatrix}\right)$. Note that $v_a$ and $v_b$ may be identical, and there are no edges of the form $(v_i, v_1)$ where $v_1$ denotes the minimal vertex. Thus $|E(G')| \leq k-1$, so $|E(G)|\leq 2k-2$.
\end{proof}

The next result is analogous to the ratio bound for 0-1 matrices in Theorem \ref{main01}, except rows and columns are replaced by the parts of a bipartite ordered graph.

\begin{thm}
Any minimally non-linear bipartite ordered graph with $k$ vertices in one part has at most $4k-2$ vertices in the other part.
\end{thm}

\begin{proof}
Given a minimally non-linear bipartite ordered graph $G$, without loss of generality assume that the first part $U$ has $k$ nodes. For each node $v_i$ in the second part $V$, we choose a neighbor in the first part using a process analogous to the one that we used for 0-1 matrices: if $v_i$ has only one neighbor then pick it, otherwise pick the smallest neighbor different from what we pick for $v_{i-1}$. 

Now we get a sequence with $k$ distinct elements without any repetition of length more than $2$ because otherwise $G$ is not minimally non-linear. The sequence cannot be longer than $2\Ex(abab,k)=4k-2$, or else it would contain some element in
$G_o\left(\begin{bmatrix}
1 & 0 & 1 & 0 \\
0 & 1 & 0 & 1
\end{bmatrix}\right)$, which has all elements non-linear.
\end{proof}

Next we obtain an upper bound of $k-1$ on the number of edges in minimally non-linear bipartite ordered graphs with $k$ vertices unless the underlying graph is $K_{2,2}$. This bound is half the upper bound for minimally non-linear ordered graphs in Theorem \ref{allord}. The lemma that we use to obtain this bound is analogous to Lemma \ref{01edge}, which we used to bound the number of ones in minimally non-linear 0-1 matrices.

\begin{lem}
The number of edges in any minimally non-linear bipartite ordered graph with $w$ vertices in one part and $h$ vertices in the other part, besides ordered graphs whose underlying graph is $K_{2,2}$, is at most $w+h-1$. 
\end{lem}

\begin{proof}
The result is clear if $G$ is an element of $G_{o}\left(
\begin{bmatrix}
1 & 0 & 1\\
0 & 1 & 1
\end{bmatrix}
\right)$, so suppose that $G$ is a minimally nonlinear bipartite ordered graph that is not an element of $G_{o}\left(
\begin{bmatrix}
1 & 0 & 1\\
0 & 1 & 1
\end{bmatrix}
\right)\cup G_{o}\left(
\begin{bmatrix}
1 & 1\\
1 & 1
\end{bmatrix}
\right)$. For each node $u\in U$, remove the edge $(u,v)\in E(G)$ with the smallest possible $v\in V$, no matter whether $u>v$ or $u<v$. So we remove exactly $|U|$ edges. 

Each $v_k\in V$ in the resulting graph $G'$ has at most one neighbor. If it has more, say $(u_a,v_k), (u_b,v_k)$, then there are $v_i$ and $v_j$, which could be identical, such that $v_i<v_k$, $v_j<v_k$ and $(u_a,v_i)\in E(G)$, $(u_b,v_j)\in E(G)$. Clearly $G$ contains some element in $G_{o}\left(
\begin{bmatrix}
1 & 0 & 1\\
0 & 1 & 1
\end{bmatrix}
\right)\cup G_{o}\left(
\begin{bmatrix}
1 & 1\\
1 & 1
\end{bmatrix}
\right)$, a contradiction. So $|E(G)|\leq |U|+|V|-1=|V(G)|-1$.
\end{proof}

\begin{cor}
The number of edges in any minimally non-linear bipartite ordered graph with $k$ total vertices is at most $k-1$ unless the underlying graph is $K_{2,2}$, and the number of edges in any minimally non-linear bipartite ordered graph with $k$ vertices in one part is at most $5k-3$. 
\end{cor}

\begin{cor}
For $k > 2$, the number of minimally non-linear bipartite ordered graphs with $k$ nodes in one part is at most $\sum_{i=\lceil (k+2)/4 \rceil}^{4k-2}\binom{k+i}{k}(i^{k}-(i-1)^{k})k^{i-1}$.
\end{cor}

\section{Open Problems}

We proved bounds for the following problems, but none of these problems are completely resolved.

\begin{enumerate}

\item 
\begin{enumerate}
\item For each $k > 0$, what is the maximum possible length of a minimally non-linear sequence with $k$ distinct letters? 
\item How many minimally non-linear sequences have $k$ distinct letters?
\item Characterize all minimally non-linear sequences with $k$ distinct letters.
\end{enumerate}

\item 
\begin{enumerate}
\item What is the maximum possible ratio between the length and width of a minimally non-linear 0-1 matrix? 
\item For each $k > 0$, what is the maximum possible number of columns in a minimally non-linear 0-1 matrix with $k$ rows? 
\item What is the maximum possible number of ones in a minimally non-linear 0-1 matrix with $k$ rows? 
\item How many minimally non-linear 0-1 matrices have $k$ rows?
\item Characterize all minimally non-linear 0-1 matrices with $k$ rows.
\end{enumerate}

\item
\begin{enumerate}
\item What is the maximum possible ratio between the sizes of the parts of a minimally non-linear bipartite ordered graph? 
\item For each $k > 0$, what is the maximum possible number of vertices in the second part in a minimally non-linear bipartite ordered graph with $k$ vertices in the first part? 
\item What is the maximum possible number of edges in a minimally non-linear bipartite ordered graph with $k$ total vertices? 
\item What is the maximum possible number of edges in a minimally non-linear bipartite ordered graph with $k$ vertices in one part? 
\item How many minimally non-linear bipartite ordered graphs have $k$ vertices in one part?
\item Characterize all minimally non-linear bipartite ordered graphs with $k$ vertices in one part.
\end{enumerate}

\item
\begin{enumerate}
\item For each $k > 0$, what is the maximum possible number of edges in a minimally non-linear ordered graph with $k$ vertices?
\item Characterize all minimally non-linear ordered graphs with $k$ vertices.
\end{enumerate}
\end{enumerate}

\section{Acknowledgments}
CrowdMath is an open program created by the MIT Program for Research in Math, Engineering, and Science (PRIMES) and Art of Problem Solving that gives high school and college students all over the world the opportunity to collaborate on a research project. The 2016 CrowdMath project is online at http://www.artofproblemsolving.com/polymath/mitprimes2016. The authors thank Gabor Tardos for proving that if $G'$ is an ordered graph obtained from $G$ by adding an edgeless vertex, then $\ex_{<}(n, G') = O(\ex_{<}(n, G)+n)$.


\begin{thebibliography}{7}
\bibitem{1} P. K. Agarwal and M. Sharir, Davenport-Schinzel Sequences and Their Geometric Applications, \emph{Cambridge University Press}, Cambridge (1995).
\bibitem{3} H. Davenport and A. Schinzel, A combinatorial problem connected with differential equations, \emph{American J. Math.} 87 (1965) 684--694.
\bibitem{5} J. Fox, J. Pach, A. Suk, The number of edges in k-quasiplanar graphs. \emph{SIAM Journal of Discrete Mathematics} 27 (2013) 550--561.
\bibitem{Fure} Z. F\"{u}redi, The maximum number of unit distances in a convex n-gon, \emph{J. Combin. Theory Ser. A} 55 (1990) 316--320.
\bibitem{FH} Z. F\"{u}redi and P. Hajnal, Davenport-Schnizel theory of matrices, \emph{Discrete Mathematics} 103 (1992) 233--251.
\bibitem{G} J. Geneson, Extremal functions of forbidden double permutation matrices, \emph{J. Combin. Theory Ser. A} 116 (2009) 1235--1244.
\bibitem{Ke} B. Keszegh, On linear forbidden submatrices, \emph{J. Combin. Theory Ser. A} 116 (2009) 232--241.
\bibitem{7} M. Klazar, On the maximum lengths of Davenport-Schinzel sequences, \emph{Contemporary trends in disc. math.} 49 (1999) 169--178.
\bibitem{KM} M. Klazar and A. Marcus, Extensions of the linear bound in the F\"{u}redi-Hajnal conjecture, \emph{Adv. in Appl. Math.} 38 (2006) 258--266.
\bibitem{MT} A. Marcus and G. Tardos, Excluded permutation matrices and the Stanley-Wilf conjecture, \emph{J. Combin. Theory Ser. A} 107 (2004) 153--160.
\bibitem{Mit} J. Mitchell, Shortest rectilinear paths among obstacles, \emph{Department of Operations Research and Industrial Engineering Technical Report} 739, Cornell University, Ithaca, NY (1987).
\bibitem{8} G. Nivasch, Improved bounds and new techniques for Davenport-Schinzel sequences and their generalizations, \emph{J. ACM} 57 (2010).
\bibitem{PT} J. Pach and G. Tardos, Forbidden paths and cycles in ordered graphs and matrices, \emph{Israel J. Math.} 155 (2006) 309--334.
\bibitem{Pet} S. Pettie, Degrees of Nonlinearity in Forbidden 0-1 Matrix Problems, \emph{Discrete Mathematics} 311 (2011) 2396--2410.
\bibitem{petmnl} S. Pettie, Generalized Davenport-Schinzel sequences and their 0-1 matrix counterparts, \emph{J. Combin. Theory Ser. A } 118 (2011) 1863--1895.
\bibitem{10} A. Suk and B. Walczak, New bounds on the maximum number of edges in k-quasi-planar graphs. International Symposium on Graph Drawing 21 (2013) 95--106.
\bibitem{gtar} G. Tardos, On 0-1 matrices and small excluded submatrices, \emph{J. Combin. Theory Ser. A } 111 (2005) 266--288.
\bibitem{gtardos} G. Tardos, Personal communication (2016)
\bibitem{weid} C. Weidert, Extremal problems in ordered graphs, CoRR abs/0907.2479 (2009).
\end{thebibliography}
\end{document}